\documentclass[letterpaper, 10 pt, conference]{ieeeconf} 
\pdfoutput=1
\IEEEoverridecommandlockouts
\usepackage[utf8]{inputenc}
\usepackage{lipsum}
\usepackage{amsmath,amssymb}
\usepackage{color}
\usepackage[linesnumbered,ruled,vlined,titlenumbered]{algorithm2e}
\let\labelindent\relax
\usepackage{enumitem}
\usepackage{tikz}
\usepackage{pgfplots}
\pgfplotsset{compat=newest}
\usepackage{color}

\usepackage{theorem}
\newtheorem{theorem}{Theorem}

\newtheorem{corollary}{Corollary}
\newtheorem{problem}{Problem}

\def\ve#1{{\mathchoice{\mbox{\boldmath$\displaystyle #1$}}%
              {\mbox{\boldmath$\textstyle #1$}}%
              {\mbox{\boldmath$\scriptstyle #1$}}%
              {\mbox{\boldmath$\scriptscriptstyle #1$}}}}

\definecolor{shadecolor}{gray}{.65}
\definecolor{institut_color_orig}{rgb}{0.63922,0.14902,0.21961}
\definecolor{institut_color_gray}{rgb}{0.5,0.14902,0.21961}
\definecolor{institut_color}{rgb}{0.85922,0.12902,0.21961}
\definecolor{my_green}{rgb}{0.2,0.75,0.35}
\definecolor{my_gray}{rgb}{0.5,0.14902,0.21961}
\definecolor{my_blue}{rgb}{0.2,0.1,0.6}
\definecolor{my_yellow}{rgb}{1,0.6,0.02}
\newcommand{\unequal}[1]{\textcolor{institut_color_orig}{#1}}
\newcommand{\blue}[1]{\textcolor{my_blue}{#1}}

\newcommand{\Code}{\mathcal{C}}
\renewcommand{\H}{\ve{H}}
\newcommand{\zeromat}{\ve{0}}
\renewcommand{\r}{\ve{r}}
\renewcommand{\c}{\ve{c}}
\renewcommand{\u}{\ve{u}}
\renewcommand{\v}{\ve{v}}
\newcommand{\e}{\ve{e}}
\newcommand{\h}{\ve{h}}
\newcommand{\htilde}{\tilde{\ve{h}}}
\newcommand{\Horth}{\ve{\widetilde{H}}}

\newcommand{\Hnotorth}{\ve{\overline{H}}}

\newcommand{\OC}[1]{\mathrm{OC}\left(#1\right)}

\newcommand{\F}{\mathbb{F}_2}
\newcommand{\V}{\mathcal{V}}
\renewcommand{\S}{\mathcal{S}}

\newcommand{\E}{\mathrm{E}}
\newcommand{\Var}{\mathrm{Var}}

\newcommand{\nrows}[1]{m_{#1}}
\newcommand{\nrank}[1]{rk_{#1}}
\newcommand{\Horig}{\H^{(\mathrm{I})}}

\renewcommand{\V}{\mathcal{V}}

\newcommand{\RS}[1]{\langle #1 \rangle}
\newcommand{\kmax}{k_\mathrm{max}}
\newcommand{\rI}{\r_\mathrm{I}}

\DeclareMathOperator{\rank}{rank}

\newcommand{\AttackerUncertainty}{\mathrm{H}(\r|\Code)}

\newcommand{\CondUncertaintyH}{\mathrm{H}(\r|\Code)}

\title{\LARGE \bf
Constructing an LDPC Code Containing a Given Vector
}

\author{Sven M\"uelich$^{1}$, Sven Puchinger$^{2}$ and Martin Bossert$^{1}$%
\thanks{$^{1}$Sven M\"uelich and Martin Bossert are with Institute of Communication Engineering, Ulm University, Germany. {\tt\small \{sven.mueelich,martin.bossert\}@uni-ulm.de}}
\thanks{$^{2}$Sven Puchinger is with the Institute for Communications Engineering, Technical University Munich,
        Germany.
        {\tt\small sven.puchinger@tum.de} This work was partly done while Sven Puchinger was with Ulm University.}%
}
\usepackage{flushend}

\begin{document}

\maketitle
\thispagestyle{empty}
\pagestyle{empty}

\begin{abstract}
		The coding problem considered in this work is to construct a linear code $\Code$ of given length $n$ and dimension $k<n$ such that a given binary vector $\r \in \F^{n}$ is contained in the code.	
		We study a recent solution of this problem by M\"uelich and Bossert, which is based on LDPC codes.
		We address two open questions of this construction.
		First, we show that under certain assumptions, this code construction is possible with high probability if $\r$ is chosen uniformly at random.
		Second, we calculate the uncertainty of $\r$ given the constructed code $\Code$.
		We present an application of this problem in the field of Physical Unclonable Functions (PUFs).
\end{abstract}

\section{PROBLEM STATEMENT}
The problem considered in this work is stated as follows:
\begin{problem}
	\label{problem}
Given a vector $\r \in \F^n$ and a desired dimension $k<n$.
Find a binary linear code $\Code(n, \approx k)$ with $\r \in \Code$.
\end{problem}

Beside the theoretical interest of the problem, a practical application is error correction for Physical Unclonable Functions (PUFs), which we will explain in Section~\ref{sec:application}.
A construction of binary linear codes based on good LDPC codes that fulfill the requirements stated in Problem~\ref{problem} was proposed by M\"uelich and Bossert in  \cite{muelich2017new}.
By \emph{good}, we mean that the resulting code can be decoded well (e.g., its Tanner graph has a large girth). In particular, this excludes constructions that randomly choose low-weight parity checks in the dual code of $\r$.
In this paper, we study two questions which have not been addressed so far.

\begin{itemize}
	\item[(P1)] Calculation of $P(\exists\Code: \r\in\Code)$, i.e., the probability that such a construction is possible for a given $\r \in U(\F^{n})$, where $U(\F^{n})$ denotes the uniform distribution of all vectors in $\F^{n}$.
	\item[(P2)] Calculation of $H(\r|\Code)$, i.e., the uncertainty of an attacker who wants to recover $\r$ and knows the code $\Code$. (The application in Section~\ref{sec:security} provides an example where $\r$ can be interpreted as secret key.)
\end{itemize}

(P1) is studied in Section~\ref{sec:successful_construction}. 
We show that under certain assumptions the construction is possible with high probability.
Numerical results indicate that these assumptions are met.
(P2) is discussed in Section~\ref{sec:security}.
We derive an analytic closed-form expression of the uncertainty and show that if the code construction is successful, we have $k-2 \leq H(\r \mid \Code) \leq k$.

\section{PRELIMINARIES}
\label{sec:preliminaries}

\subsection{Code Construction}
\label{subsec:construction}
The code construction solving Problem~\ref{problem} which was proposed in \cite{muelich2017new} is based on LDPC codes  \cite{gallager1962low}.
LDPC codes are constructed by generating a parity check matrix of low density.
Let $\Horig$ be a low-density $\nrows{\Horig} \times n$ matrix of rank $\nrank{\Horig}$, where $\nrows{\Horig}$ is allowed to be larger than $\nrank{\Horig}$.
$\Horig$ can be interpreted as a decoding matrix of an LDPC code of length $n$ and dimension $n-\nrank{\Horig}$.
Assume $\Horig$ is given, below (cf. Section~\ref{subsec:obtaindecodingmatrix}) we describe how to obtain $\Horig$.
Table~\ref{tab:size_examples} shows for three examples how much larger $\nrows{\Horig}$ can be compared to $\nrank{\Horig}$.

\begin{table}[hbtp]
	\begin{center}
		\def\arraystretch{1.0}
		\setlength{\tabcolsep}{0.1cm}
		\begin{tabular}{ c | c | c | c }		
			$n$ & $k$ & $\nrows{\Horig}$ &  $\nrank{\Horig}$ \\
			\hline
			128 & 13 & 881 & 115 \\
			128 & 56 & 349 & 72\\
			256 & 106 & 555 & 150
		\end{tabular}
	\end{center}
	\caption{Code examples \cite{muelich2017new}.}
	\label{tab:size_examples}
\end{table}

The aim is to find an LDPC code $\Code$ such that $\r \in \Code$, where $\r \in U(\F^{n})$ is given.
Algorithm~\ref{alg:scheme} summarizes the method from \cite{muelich2017new} for constructing a convenient decoding matrix $\H$ from $\r$ and $\Horig$:
Since we want to have $\r\in\Code$, it is required that $\H\r^T = \zeromat$, and hence $\h_i \r^T = 0$ for all rows $\h_i$ of $\H$.
We select the rows $\h_i$ from the matrix $\Horig$ for which this condition is fulfilled, i.e., the rows which are orthogonal\footnote{In coding theory often the term dual is used instead of orthogonal. In this work these two terms are exchangeable.} to $\r$, and use them to construct the decoding matrix $\H$ we are aiming for.
We require that $k = dim(\Code) \geq H(\r|\Code)$.
Hence $\H$ should have at most rank $n-H(\r|\Code)$.
While adding parity check equations to $\H$, the rank of the matrix increases, and hence the dimension $k$ decreases. 
We stop the process, when $\H$ has a desired dimension $k$ in order to be flexible in adjusting the code rate of the final LDPC code.

\begin{algorithm}[hbt]
	\KwIn{Vector $\r \in \F^n$, dec. matrix $\Horig$ with $\nrows{\Horig}$ rows}
	\KwOut{Decoding matrix $\H$ of a code $\Code(n,k)$, such that $\r \in \Code$  }
	$\H$ is matrix with $n$ columns and $0$ rows.\label{alg:scheme2}\\
	\For{$i=1,2,\dots,\nrows{\Horig}$}{
		\If{$i$-th row $\h_i$ of $\Horig$ is orthogonal to $\r$}{
			Vertically append $h_i$ to $\H$.
		}
	}
	\Return{$\H$} 
	\caption{LDPC code construction algorithm \cite{muelich2017new}}
	\label{alg:scheme}
\end{algorithm}

\subsection{How to get a decoding matrix $\Horig$}
\label{subsec:obtaindecodingmatrix}

In Section~\ref{subsec:construction} it was assumed that $\Horig$ is given.
This section discusses how $\Horig$ can be constructed.
In \cite{muelich2017new}, the rows of $\Horig$ were chosen as the union of rows of different decoding matrices $\Horig_i$ of various constructions of LDPC codes $\Code_i$ of length $n$ (Euclidean and Projective Geometries \cite{kou2000low}, and Reed--Solomon based \cite{djurdjevic2003class} LDPC codes).
As described in Section~\ref{subsec:construction}, we are only interested in parity check equations that are orthogonal to $\r$.
In general, the error correction performance of a code decreases when the number of parity check equations decreases.
Hence, if we use only one LDPC construction method and take the appropriate rows from the corresponding parity check matrix, we result in a weak error correction performance. 
Thus, to increase the number of errors that can be corrected, it is necessary to extend the number of parity check equations what can be done by combining selected rows of several LDPC parity check matrices.

\subsection{Notation}

Let $\V \subseteq \F^{n}$ be a subspace. 
Then, 
\begin{align*}
\OC{\V} := \left\{ \u \in \F^n : \u\v^T=0 \; \forall \v \in \V \right\}
\end{align*}
with $\dim(\OC{\V}) = n-\dim(\V)$ is the \emph{orthogonal complement} of $\V$. 
Let $\H$ be a matrix with \emph{row space} $\RS{\H}$.
For convenience, we write $\OC{\H} := \OC{\RS{\H}}$.
In coding theory, the orthogonal complement of a code is often called \emph{dual code}.
A \emph{parity check matrix} of a linear code $\Code$ is a basis of $\OC{\Code}$.
Any other $\H$ whose rows span $\OC{\Code}$ is called \emph{decoding matrix}.

\section{SUCCESSFUL CODE CONSTRUCTIONS}
\label{sec:successful_construction}

In this section, we would like to analyze for which vectors $\r$ it is possible to construct an LDPC code $\Code$ of sufficiently small dimension $k$ that contains $\r$ as codeword.
For instance, if $k$ should be at most $\kmax$, we need output $\H$ of Algorithm~\ref{alg:scheme} to have at least rank $n-\kmax$.

\subsection{Theoretical Ideas}
\label{subsec:code_construction_theoretical}

Let $\h_i$ be the $i$-th row of $\Horig$.
The probability of $\h_i$ being orthogonal to $\r$ depends on the number of ones in the positions of $\r$ which are indexed by the support of $\h_i$. Since $\r \sim \mathcal{U}(\F^n)$, it follows that
\begin{align*}
\Pr\left( \r \cdot \h_i^T=0  \right) = \frac{1}{2}.
\end{align*}
Due to sparsity, the support of two distinct rows $\h_i \neq \h_j$ of $\Horig$ is unlikely to overlap.
In case of an overlap, it only concerns a small number of positions. 
Therefore, we assume that the events $\r \cdot \h_i^T=0$ are statistically independent.
Thus, the number of rows $\nrows{\H}$ of $\H$ is binomially distributed with parameters $\nrows{\Horig}$ and $\tfrac{1}{2}$, i.e.,
$\nrows{\H} \sim \mathrm{Bin}(\nrows{\Horig}, \tfrac{1}{2})$.
We will show in Section~\ref{subsec:code_construction_practical} that this can be verified for practical examples.

We can thus conclude that if the decoding matrix $\Horig$ contains many more rows than its rank, i.e., $\nrows{\Horig} \gg \nrank{\Horig}$, then the probability of having less than $\nrank{\Horig}$ rows in $\H$ is negligible.
Also, the rank of $\H$ fulfills
\begin{align*}
\rank(\H) \geq \rank(\Horig)-1
\end{align*}
with large probability since the $\nrows{\H} \gg \nrank{\Horig}$ rows of $\H$ are very likely to contain a generating set of
\begin{align*}
\V := \RS{\Horig} \cap \OC{\r}
\end{align*}
and
\begin{align*}
	\dim\left(\V\right) &= \dim\left(\RS{\Horig}\right)+\dim\left(\OC{\r}\right)\\&-\underset{\leq n}{\underbrace{\dim\left(\RS{\Horig}+\OC{\r}\right)}} \\
	&\geq \rank\left(\Horig\right)-1.
\end{align*}

\subsection{Numerical Results}
\label{subsec:code_construction_practical}

We verify the results in Section~\ref{subsec:code_construction_theoretical} in a practical scenario.
The example uses an $(512,139)$ Euclidean Geometry code, whose decoding matrix $\Horig$ has $\nrows{\Horig} = 4672$ rows.
We generated random vectors $\r \sim \mathcal{U}(\F^n)$ and used Algorithm~\ref{alg:scheme} to obtain a code with decoding matrix $\H$ containing $\r$.
Using the results of Section~\ref{subsec:code_construction_theoretical}, the number of rows $\nrows{\H}$ should be approximately $\mathrm{Bin}(\nrows{\Horig}=4672,0.5)$ distributed.
Figure~\ref{fig:empirical_cdf_nrowsH} visualizes the empirical cdf (sample size = $10^6$) in comparison to the theoretical cdf.

\begin{figure}[h!]
	\begin{center}
		\input{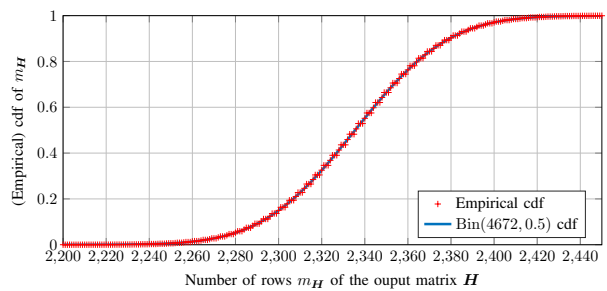}
	\end{center}
	\caption{Comparison of the empirical cdf of the number of rows $\nrows{\H}$ of $\H$ with the cdf of a $\mathrm{Bin}(\nrows{\Horig},0.5)$ distributed random variable. Sample size $10^6$, number of rows $\nrows{\Horig} = 4672$ of $\Horig$.}
	\label{fig:empirical_cdf_nrowsH}
\end{figure}

It can be seen that the two curves almost coincide.
Estimated mean and variance of $\nrows{\H}$ are $\widehat{\mu} \approx 2335.99$ and $\widehat{\sigma^2} \approx 1169.49$. This is close to the theoretical values $\E(\nrows{\H}) = \tfrac{\nrows{\Horig}}{2} = 2336$ and $\Var(\nrows{\H}) = \tfrac{\nrows{\Horig}}{4} = 1168$.
In addition, it was also observed that $\rank(\H)= 272 = \rank(\Horig)-1$ for all tested cases (sample size $=10^3$ due to larger complexity of rank computation), which also coincides with the theoretical prediction.

\section{ENTROPY DISCUSSION}
\label{sec:security}

Assume $\r$ is a secret that should be hidden from an attacker that knows the code $\Code$. 
This is, for instance, the case in the application presented in Section~\ref{sec:application}.
Since the code $\Code$ contains $2^k$ codewords, including $\r$, the uncertainty of $\r$ for an attacker is at most $k$.
In this section, we derive a lower bound on the uncertainty and show that in most cases, it is lower-bounded by $k-2$.

\subsection{Uncertainty of an Attacker}
\label{subsec:uncertainty_attacker}

Let $\Horig$ denote the original decoding matrix and $\Horth$ be the rows of $\Horig$ that are orthogonal to $\r$ and $\Hnotorth$ be the other rows.
The ``attack'' described above, which uses the fact that $\r$ is a codeword of $\Code$, or in other words
\begin{align}
	\r \in \Code := \OC{\Horth}, \label{eq:attack_cond_1}
\end{align}
is suboptimal since it does not make use of $\Hnotorth$. We can in addition use the latter by the following observation:
\begin{align}
	\r \in \bigcap_{i=1}^{\nrows{\Hnotorth}} \V_i, \label{eq:attack_cond_2}
\end{align}
where $\V_i := \F^n \setminus \OC{\h_i}$ with the $i$-th row $\h_i$ of $\Hnotorth$.
By combining Equations \eqref{eq:attack_cond_1} and \eqref{eq:attack_cond_2}, we obtain
\begin{align*}
\r \in \Code \cap \left( \bigcap_{i=1}^{\nrows{\Hnotorth}} \V_i \right) =: \S.
\end{align*}
The uncertainty of an attacker $\AttackerUncertainty$ about $\r$ is directly determined by the size of $\S$.
More precisely, since $\r \in \mathcal{U}(\F^n)$, the uncertainty is lower-bounded by
\begin{align*}
\AttackerUncertainty \geq \log_2(|\S|).
\end{align*}
We use Theorem~\ref{thm:rI_H'} to derive a lower bound on $\AttackerUncertainty$.

\begin{theorem}\label{thm:rI_H'} 
	A response $\r$ fulfills Conditions \eqref{eq:attack_cond_1} and \eqref{eq:attack_cond_2} if and only if
	$\r' \in \OC{\H'}$, where
	\begin{align*}
	\r' := [\r,1] \in \F^{n+1}, 
	\end{align*}
	\begin{align*}
	\H' := \begin{bmatrix}
	\Horth & \ve{0}_{\nrows{\Horth} \times 1} \\
	\Hnotorth & \ve{1}_{\nrows{\Hnotorth} \times 1}
	\end{bmatrix} \in \F^{(\nrows{\Horth}+\nrows{\Hnotorth}) \times (n+1)}.
	\end{align*}
\end{theorem}

\begin{proof}
	Condition \eqref{eq:attack_cond_1} is fulfilled if and only if
	\begin{align*}
	\h_i \cdot \r^T = 0 \quad \forall i,
	\end{align*}
	where $\h_i$ is the $i$-th row of $\Horth$. This again holds if and only if
	\begin{align*}
	[\h_i,0] \cdot [\r,1]^T = 0 \quad \forall i,
	\end{align*}
	where $[\h_i,0]$ is the $i$-th row of $\begin{bmatrix} \Horth & \ve{0}_{\nrows{\Horth} \times 1} \end{bmatrix}$, which means that $[\r,1]$ is in the right kernel of $\begin{bmatrix} \Horth & \ve{0}_{\nrows{\Horth} \times 1} \end{bmatrix}$.
	Similarly, Condition \eqref{eq:attack_cond_2} is satisfied if and only if
	\begin{align*}
	\h_i \cdot \r^T = 1 \quad \forall i,
	\end{align*}
	where $\h_i$ is the $i$-th row of $\Hnotorth$. This implies that
	$[\h_i,1] \cdot [\r,1]^T = 0 \quad \forall i$,
	where $[\h_i,1]$ is the $i$-th row of $\begin{bmatrix} \Hnotorth & \ve{1}_{\nrows{\Hnotorth} \times 1} \end{bmatrix}$, which means that $[\r,1]$ is in the right kernel of $\begin{bmatrix} \Hnotorth & \ve{1}_{\nrows{\Hnotorth} \times 1} \end{bmatrix}$.
	Since the rows of $\begin{bmatrix} \Horth & \ve{0}_{\nrows{\Horth} \times 1} \end{bmatrix}$ and $\begin{bmatrix} \Hnotorth & \ve{1}_{\nrows{\Hnotorth} \times 1} \end{bmatrix}$ are exactly those of $\H'$, we get that the conditions are fulfilled if and only if $[\r,1]$ is in the right kernel of~$\H'$.
\end{proof}

Using the statement of Theorem~\ref{thm:rI_H'}, the cardinality of $\log_2(|\S|)$, and therefore a lower bound on the attacker's uncertainty about $\r$, is directly determined by the rank of the matrix $\H'$ in the following way.
\begin{corollary}\label{cor:log2S}
	$\CondUncertaintyH \geq n-\rank\left( \H' \right)$.
\end{corollary}

\begin{proof}
	Using the rank nullity theorem $(\ast)$, we get
	\begin{align*}
		\log_2(|\S|) &= \dim\left( \OC{\H'} \right)-1 \\
		&\overset{(\ast)}{=} (n+1)-\rank\left( \H' \right)-1 \\
		&= n-\rank\left( \H' \right),
	\end{align*}
	which proves the claim. 
\end{proof}

Since the matrix $\H'$ is directly known after constructing $\Code$, Corollary~\ref{cor:log2S} provides a tool to determine a lower bound on the uncertainty of the attacker.
We will see in the next subsection that $\log_2(|\S|)$ is at least $k-2$ with large probability, where $k$ is the dimension of $\Code$.

\subsection{Practical Considerations}
\label{subsec:prevention}

In Section~\ref{subsec:uncertainty_attacker}, we have derived a lower bound on an attacker's uncertainty about $\r$ which only depends on the rank of the matrix $\H'$.
The following theorem shows that this rank is sufficiently small for the case of a successful code construction, i.e., $\rank(\Horth) \geq \rank(\Horig)-1$.

\begin{theorem}
	If $\rank(\Horth) \geq \rank(\Horig)-1$ then $\rank(\H') \leq \rank(\Horth)+2$.
\end{theorem}

\begin{proof}
	Due to $\rank(\Horth) \geq \rank(\Horig)-1$, we get
	\begin{align*}
		\dim\left( \RS{\Horth} \cap \RS{\Hnotorth} \right) &= \underset{\geq \, \rank(\Horig)-1}{\underbrace{\rank\left(\Horth\right)}} \\ &+ \rank\left(\Hnotorth\right) - \underset{= \, \rank(\Horig)}{\underbrace{\dim\left( \RS{\Horth} + \RS{\Hnotorth} \right)}} \\
		&\geq \rank\left(\Hnotorth\right)-1.
	\end{align*}
	Hence, since $\RS{\Horth} \cap \RS{\Hnotorth}$ is a subspace of $\RS{\Hnotorth}$ of dimension at least $\rank\left(\Hnotorth\right)-1$, the vector space $\RS{\Hnotorth}$ is a direct sum of the form
	\begin{align*}
	\RS{\Hnotorth} = \RS{\h} + \left(\RS{\Horth} \cap \RS{\Hnotorth}\right),
	\end{align*}
	where $\h \in \RS{\Hnotorth}$.
	This means that all rows of the lower half of $\H'$, i.e., $\begin{bmatrix} \Hnotorth & \ve{1}_{\nrows{\Hnotorth} \times 1} \end{bmatrix}$, are of the form
	\begin{enumerate}[label=(\roman*)]
		\item $[\h+\htilde,1]$, where $\htilde$ is in the span of the rows of $\Horth$.
		\item $[\htilde,1]$, where $\htilde$ is in the span of the rows of $\Horth$.
	\end{enumerate}
	If there is at least one row $\h'_1 = [\h+\htilde',1]$ of type $\mathrm{(i)}$, then all such rows are in the span of $\begin{bmatrix} \Horth & \ve{0}_{\nrows{\Horth} \times 1} \end{bmatrix}$ and $\h'_1$ since
	\begin{align*}
		[\h+\htilde,1] = \underset{= \, \h'_1}{\underbrace{[\h+\htilde',1]}} + \underset{\in \, \RS{\begin{bmatrix} \Horth & \ve{0}_{\nrows{\Horth} \times 1} \end{bmatrix}}}{\underbrace{[\htilde-\htilde',0].}}
	\end{align*}
	A similar argument holds for rows of type $\mathrm{(ii)}$, if there is one such row $\h'_2$.
	
	Hence, the rows of $\H'$ are spanned by the rows of the matrix $\RS{\begin{bmatrix} \Horth & \ve{0}_{\nrows{\Horth} \times 1} \end{bmatrix}}$ and $\h'_1$ and $\h'_2$ and its rank is
	\begin{align*}
	\rank\left( \H' \right) \leq \rank\left( \begin{bmatrix} \Horth & \ve{0}_{\nrows{\Horth} \times 1} \end{bmatrix} \right) + 2 = \rank\left(\Horth\right) +2,
	\end{align*}
	which proves the claim.
\end{proof}

Thus, in the case of a successful code construction, by Corollary~\ref{cor:log2S}, the uncertainty of the attacker about the vector $\r$ is lower-bounded as follows.
\begin{corollary}
	If $\rank(\Horth) \geq \rank(\Horig)-1$, then $\AttackerUncertainty \geq k-2$.
\end{corollary}

Recall that we have shown in Section~\ref{sec:successful_construction}, that the assumption $\rank(\Horth) \geq \rank(\Horig)-1$ is fulfilled with high probability and therefore, with high probability, we get $\AttackerUncertainty \geq k-2$.

\section{APPLICATION: PHYSICAL UNCLONABLE FUNCTIONS (PUFS)}
\label{sec:application}

The problem discussed in this paper can be found in the context of Physical Unclonable Functions (PUFs), devices that can be used for cryptographic purposes like identification, authentication and secure key generation. 
A PUF extracts a unique and reproducible sequence of bits, which is called response, from an integrated circuit.
The uniqueness results from randomness which is intrinsic to each device, due to technical and physical variations within the manufacturing process. 
Since these variations cannot be controlled by the manufacturer, devices which exploit this behavior are called physically unclonable, since the behavior cannot be cloned.
Usually either the delay behavior of a device (e.g. Ring-Oscillator PUFs \cite{suh2007physical}) or the initialization behavior of memory cells (e.g. SRAM PUFs \cite{guajardo2007fpga}) is used in order to extract responses.
An extracted response can for example be used as cryptographic key due to its uniqueness.

One of the major advantages of using PUFs for key generation is, that there is no need to store the key in a non-volatile memory, since it can be simply reproduced when it is needed by the cryptosystem. 
Storing a key in a non-volatile memory makes a system vulnerable to physical attacks, even when a protected memory is used.
Due to measurement noise or environmental conditions like temperature, supply voltage or aging of the chip, errors occur when the response is repeatedly extracted.
For this reason, error correction within a so-called \emph{helper data algorithm} is needed.
For comprehensive details regarding different PUF constructions we refer to the literature \cite{bohm2012physical,maes2013physically,wachsmann2014physically}.

\subsection{Existing Helper Data Algorithms}

Different helper data algorithms have been proposed in the literature.
The most popular ones are \emph{Code-Offset Construction} and \emph{Syndrome Construction} \cite{linnartz2003new,dodis2004fuzzy}.
Both have in common, that an error correcting code as well as additional helper data are required in order to reproduce a PUF response.
Helper data algorithms consist of two phases, \emph{initialization phase} and \emph{reproduction phase}.
In the initialization phase, helper data are extracted from an initial PUF response. 
This phase is performed only once within a secure environment during the manufacturing process of the device. 
In the reproduction phase, these helper data are used in order to reproduce the initial response from a regenerated, erroneous response. 
The reproduction phase occurs in field whenever the cryptosystem needs access to the key.

We describe the Code-Offset Construction according to the literature \cite{linnartz2003new,dodis2004fuzzy}, which is visualized in Figure~\ref{fig:codeoffset}.
In the initialization phase of the Code-Offset Construction, a codeword $\c$ from a specified code $\Code$ is randomly chosen and added bitwise to the initial response $\r$. 
The model assumes that $\r \sim \mathcal{U}(\F^n)$.
The bit sequence $\h$ which is obtained by this operation is stored as helper data in a non-volatile, possibly non-protected helper data storage.
In the reproduction phase, which is executed in the field as often as the cryptosystem needs the key, the helper data $\h$ are added to a new extracted response $\r'$ in a pre-processing step of the error correction.
The result of this operation has the form $\c+\e$ where $\e$ is a vector of low weight and hence can be interpreted as error vector. 
A decoder for the chosen code can be used in order to calculate $\c$.
Adding again the helper data $\h$ to $\c$ yields the initial response $\r$ which was extracted in the initialization phase.

\begin{figure}[h]
	\scalebox{1.0}{{
\resizebox{0.48\textwidth}{!}{
\centering
\begin{tikzpicture}[scale=0.6]

\draw (0,0) -- (0,3);
\draw (0,3) -- (3,3);
\draw (3,3) -- (3,0);
\draw (3,0) -- (0,0);
\draw (1.5,1.5) node {PUF};

\draw (3,1.5) -- (4.5,1.5);
\draw[->,>=latex] (5,1.5-0.86602540378) -- (5,-3);
\draw[fill] (4.5,1.5) circle (2pt);
\draw[fill] (5.5,1.5) circle (2pt);
\draw[fill] (5,1.5-0.86602540378) circle (2pt);
\draw [dashed,domain=-80:20] plot ({4.5+cos(\x)}, {1.5+sin(\x)});
\draw[->,>=latex] (4.5,1.5) -- (5,1.5-0.86602540378);
\draw[->,>=latex] (4.5,1.5) -- (5.4,1.5);
\draw (7.3,1.5) node [above] {\unequal{$r$}};
\draw (5,-0.7) node [left] {\unequal{$r' = r \oplus e$~~~~~~}};
\draw (5,-1.3) node [left] {\unequal{$~ = c \oplus h \oplus e$}};

\draw (9,0) -- (9,3);
\draw (9,3) -- (12,3);
\draw (12,3) -- (12,0);
\draw (12,0) -- (9,0);
\draw (10.5,2.5) node {Helper};
\draw (10.5,1.5) node {Data};
\draw (10.5,0.5) node {Generation};
\draw (10.5,0) node [below] {\blue{$c \in_R \mathcal{C}$}};
\draw (10.5,-0.7) node [below] {\blue{$r = c \oplus h$}};

\draw[->,>=latex] (5.5,1.5) -- (9,1.5);
\draw[->,>=latex] (12,1.5) -- (14,1.5);
\draw (13,1.5) node [above] {\unequal{$\mathcal{C},h$}};

\draw (14,0) -- (14,3);
\draw (14,3) -- (17,3);
\draw (17,3) -- (17,0);
\draw (17,0) -- (14,0);
\draw (15.5,2.5) node {Helper};
\draw (15.5,1.5) node {Data};
\draw (15.5,0.5) node {Storage};

\draw (16,0) -- (16,-2);
\draw (16,-2) -- (6,-2);
\draw[->,>=latex] (6,-2) -- (6,-3);
\draw (6,-2.5) node [right] {\unequal{$h$}};
\draw[->,>=latex] (10.5,-2) -- (10.5,-3);
\draw (10.5,-2.5) node [right] {\unequal{$\mathcal{C}$}};

\draw (4,-3) -- (7,-3);
\draw (7,-3) -- (7,-6);
\draw (7,-6) -- (4,-6);
\draw (4,-6) -- (4,-3);
\draw (5.5,-4.0) node {Pre--};
\draw (5.5,-5.0) node {processing};
\draw (4.4,-6.0) node [below] {\blue{$y = r' \oplus h$}};
\draw (5.6,-6.6) node [below] {\blue{$~ = c \oplus h \oplus e \oplus h$}};
\draw (4.2,-7.2) node [below] {\blue{$~~ = c \oplus e$}};

\draw[->,>=latex] (7,-4.5) -- (9,-4.5);
\draw (8,-4.5) node [above] {\unequal{$y$}};

\draw (9,-3) -- (12,-3);
\draw (12,-3) -- (12,-6);
\draw (12,-6) -- (9,-6);
\draw (9,-6) -- (9,-3);
\draw (10.5,-4.5) node {Decoder};
\draw (10.5,-6.0) node [below] {\blue{$\hat{c} = dec(y)$}};


\draw[->,>=latex] (12,-4.5) -- (16,-4.5);

\draw (14,-4.5) node [above] {\unequal{$\hat{r} = \hat{c} \oplus h$}};


\draw (17,-4.5) node {Key};

\draw [dashed] (6,-1.5) -- (6,4.5);
\draw [dashed] (6,4.5) -- (18,4.5);
\draw [dashed] (18,4.5) -- (18,-1.5);
\draw [dashed] (6,-1.5) -- (18,-1.5);
\draw (9.5,4.4) node [below] {Initialization};

\draw [dashed] (2.5,-1.8) -- (2.5,-8.6);
\draw [dashed] (2.5,-1.8) -- (12.5,-1.8);
\draw [dashed] (12.5,-1.8) -- (12.5,-8.6);
\draw [dashed] (12.5,-8.6) -- (2.5,-8.6);
\draw (6,-8.8) node [above] {Reproduction};

\end{tikzpicture}
}
}}
	\caption{Code-Offset Construction \cite{linnartz2003new,dodis2004fuzzy} for generating helper data.}
	\label{fig:codeoffset}
\end{figure}
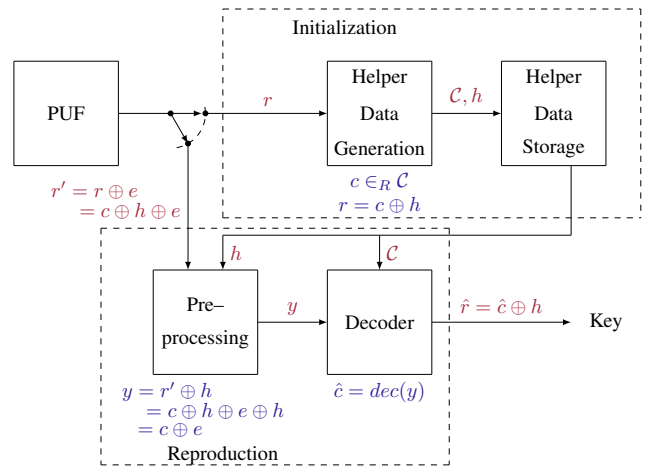

\subsection{New Helper Data Algorithm by M\"uelich and Bossert}

Recently, a scheme which only uses a code without additional helper data was proposed in \cite{muelich2017new}.
The main idea of that new scheme, which is visualized in Figure~\ref{fig:scheme_ldpc}, is to construct a code $\Code$, such that the initial PUF response $\r$ directly is a codeword of $\Code$.
Note, that here Problem~\ref{problem} occurs.
The code construction is executed in the initialization phase of the scheme and was discussed in detail in Section~\ref{subsec:construction}.
The response $\r'$, which is extracted in the reproduction phase, slightly differs from the initial response $\r$ and hence can be described as $\r' = \r + \e$, where $\e$ is again a vector of low weight.
If a decoder for the constructed code $\Code$ is used, $\r$ can be recovered by simply decoding $\r'$ (cf. Algorithm~\ref{alg:schemerep}).

\begin{figure}[h]
	\scalebox{1.0}{{
\resizebox{0.48\textwidth}{!}{
\centering
\begin{tikzpicture}[scale=0.6]

\draw (0,0) -- (0,3);
\draw (0,3) -- (3,3);
\draw (3,3) -- (3,0);
\draw (3,0) -- (0,0);
\draw (1.5,1.5) node {PUF};

\draw (3,1.5) -- (4.5,1.5);
\draw (5,1.5-0.86602540378) -- (5,-4.5);
\draw[fill] (4.5,1.5) circle (2pt);
\draw[fill] (5.5,1.5) circle (2pt);
\draw[fill] (5,1.5-0.86602540378) circle (2pt);
\draw [dashed,domain=-80:20] plot ({4.5+cos(\x)}, {1.5+sin(\x)});
\draw[->,>=latex] (4.5,1.5) -- (5,1.5-0.86602540378);
\draw[->,>=latex] (4.5,1.5) -- (5.4,1.5);
\draw (7.3,1.5) node [above] {\unequal{$r$}};
\draw  (6,-1.3) node [left] {\unequal{$r' = r \oplus e$~~~~~~}};

\draw (9,0) -- (9,3);
\draw (9,3) -- (12,3);
\draw (12,3) -- (12,0);
\draw (12,0) -- (9,0);
\draw (10.5,2.5) node {Helper};
\draw (10.5,1.5) node {Data};
\draw (10.5,0.5) node {Generation};
\draw (10.5,0) node [below] {\blue{Construct Code $\mathcal{C}$}};
\draw (10.5,-0.7) node [below] {\blue{such that $r \in \mathcal{C}$}};

\draw[->,>=latex] (5.5,1.5) -- (9,1.5);
\draw[->,>=latex] (12,1.5) -- (14,1.5);
\draw (13,1.5) node [above] {\unequal{$\mathcal{C}$}};

\draw (14,0) -- (14,3);
\draw (14,3) -- (17,3);
\draw (17,3) -- (17,0);
\draw (17,0) -- (14,0);
\draw (15.5,2.5) node {Helper};
\draw (15.5,1.5) node {Data};
\draw (15.5,0.5) node {Storage};

\draw (16,0) -- (16,-2);
\draw (16,-2) -- (10.5,-2);
\draw[->,>=latex] (10.5,-2) -- (10.5,-3);
\draw (10.5,-2.5) node [right] {\unequal{$\mathcal{C}$}};

\draw[->,>=latex] (5,-4.5) -- (9,-4.5);
\draw (8,-4.5) node [above] {\unequal{$r'$}};

\draw (9,-3) -- (12,-3);
\draw (12,-3) -- (12,-6);
\draw (12,-6) -- (9,-6);
\draw (9,-6) -- (9,-3);
\draw (10.5,-4.5) node {Decoder};
\draw (10.5,-6.0) node [below] {\blue{$\hat{r} = dec(r')$}};


\draw[->,>=latex] (12,-4.5) -- (16,-4.5);

\draw (14,-4.5) node [above] {\unequal{$\hat{r}$}};

\draw (17,-4.5) node {Key};

\draw [dashed] (6,-1.5) -- (6,4.5);
\draw [dashed] (6,4.5) -- (18,4.5);
\draw [dashed] (18,4.5) -- (18,-1.5);
\draw [dashed] (6,-1.5) -- (18,-1.5);
\draw (9.5,4.4) node [below] {Initialization};

\draw [dashed] (2.5,-1.8) -- (2.5,-8.6);
\draw [dashed] (2.5,-1.8) -- (12.5,-1.8);
\draw [dashed] (12.5,-1.8) -- (12.5,-8.6);
\draw [dashed] (12.5,-8.6) -- (2.5,-8.6);
\draw (6,-8.8) node [above] {Reproduction};

\end{tikzpicture}
}
}}
	\caption{Scheme according to \cite{muelich2017new}.}
	\label{fig:scheme_ldpc}
\end{figure}
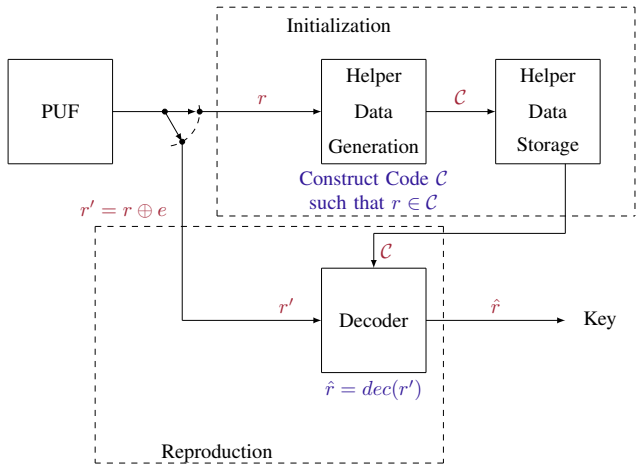

Algorithm~\ref{alg:scheme} which has been stated in Section~\ref{subsec:construction} can directly be used for the initialization phase of the Helper Data Algorithm proposed in \cite{muelich2017new}, while Algorithm~\ref{alg:schemerep} summarizes the corresponding reproduction phase.
Note, that if the uncertainty is too low, the PUF can be removed from the set of suitable PUFs. This is no problem since it is done during the manufacturing process.
\begin{algorithm}
	\KwIn{Re-extracted PUF response $\r = \rI + \e \in \{0,1\}^n$, $\H$}
	\KwOut{Decoding result $\hat{\r} = dec(\r, \H)$}
	$\hat{\r} = dec(\r,\H)$\\
	\Return $\hat{\r}$ \label{alg:schemerep2}\\
	\caption{Helper Data Algorithm: Reproduction phase \cite{muelich2017new}}
	\label{alg:schemerep}
\end{algorithm}

\subsection{Comparison to Existing Schemes}

The essential difference between the new helper data algorithm and previous ones is that only a code is needed, but no further helper data are required.
Of course, this can only be advantageous in PUF applications that do not use a challenge-response behavior.
However, there are many applications in which this requirement is met.
This, for instance, includes Physical Obfuscated Keys (POKs) \cite{gassend2003ma}.
Another application of PUFs with a single response is identification, where the response serves a unique identifier of a device, e.g., \cite{lofstrom2000ic}.\footnote{It is not suited for an authentication protocol as, e.g., described in~\cite{gassend2002silicon}, since this would require to generate and store a code for each possible challenge.}
In these cases, the new scheme is interesting since:
\begin{enumerate}
	\item There is less potential for side-channel attacks which aim on helper data, e.g. the attack described in \cite[Section 5.2]{merli2011side} is not applicable with that scheme.
	\item Processing cost in the reproduction phase is reduced, i.e., there is no need for an additional pre-processing step before decoding (cf. Figures~\ref{fig:codeoffset} and~\ref{fig:scheme_ldpc}).
	Also this allows a smaller implementation, which is essential in PUF applications as they are usually implemented on rather limited hardware.
	\item The amount of helper data is reduced in comparison to other methods based on LDPC codes.
\end{enumerate}
Moreover, it is a conceptual new method in comparison to existing schemes.
Besides PUFs, we belief that the problem also is interesting for other applications.

\section{CONCLUSION}
\label{sec:conclusion}
This work discussed the problem of constructing a code such that a given binary vector is included as codeword.
We discussed open questions about constructability and the security level of an LDPC code construction proposed in \cite{muelich2017new}. 
The results of this work essentially imply two insights:
First, the code construction is possible with high probability for desired code length and dimension.
Second, the security level is known and sufficiently large when the dimension $k$ of the code is chosen as $k \geq H(\r|\Code)+2$.
Using these results, we have shown the practicability of the helper data scheme introduced in \cite{muelich2017new}.
It is of further interest, which other code classes can be used in order to solve the stated problem and whether more applications which contain the stated problem can be identified.

\addtolength{\textheight}{-12cm}

\end{document}